\documentclass{article}
\usepackage[utf8]{inputenc}
\usepackage{amsmath}
\usepackage{amssymb}
\usepackage{graphicx}
\usepackage{amsthm}
\usepackage{latexsym}
\newtheorem{theorem}{Theorem}

\newtheorem{definition}{Definition}
\newtheorem{prop}{Proposition}
\newtheorem{lemma}{Lemma}

\usepackage{algpseudocode}
\usepackage{algorithmicx}
\usepackage{algorithm}

\newtheorem{problem}{Problem}
\newcommand{\beq}{\begin{equation}}
\newcommand{\eeq}{\end{equation}}
\newcommand{\barr}{\left[\begin{array}}
\newcommand{\earr}{\end{array}\right]}

\newcommand{\rank}{\mbox{rank}\,}

\newcommand{\fq}{\mathbb{F}_{q}}

\newcommand{\bpf}{\begin{proof}}
\newcommand{\epf}{\end{proof}}

\newcommand{\ftwo}{\ensuremath{\mathbb{F}_{2}}}

\newcommand{\ff}{\ensuremath{\mathbb{F}}}

\newcommand{\bi}{\begin{itemize}}
\newcommand{\ei}{\end{itemize}}
\newcommand{\bnum}{\begin{enumerate}}
\newcommand{\enum}{\end{enumerate}}
\newcommand{\bc}{\begin{center}}

\newcommand{\al}{\alpha}

\title{A Complete algorithm for local inversion of maps: Application to Cryptanalysis}
\author{Virendra Sule\\Dept.\ of Electrical Engineering\\Indian Institute of Technology Bombay, India\\(vrs@ee.iitb.ac.in)}
\begin{document}
\maketitle
\emph{Subject Classification}: cs.CC, cs.CR, cs.SC, cs.LO

\begin{abstract}
For a map (function) $F(x):\ftwo^n\rightarrow\ftwo^n$ and a given $y$ in the image of $F$ the problem of \emph{local inversion} of $F$ is to find all inverse images $x$ in $\ftwo^n$ such that $y=F(x)$. In Cryptology, such a problem arises in Cryptanalysis of One way Functions (OWFs). The well known TMTO attack in Cryptanalysis is a probabilistic algorithm for computing one solution of local inversion using $O(\sqrt N)$ order computation in offline as well as online for $N=2^n$. This paper proposes a complete algorithm for solving the local inversion problem which uses linear complexity for a unique solution in a periodic orbit. The algorithm is shown to require an offline computation to solve a hard problem (possibly requiring exponential computation) and an online computation dependent on $y$ that of repeated forward evaluation $F(x)$ on points $x$ in $\ff_{2^n}$ which is polynomial time at each evaluation. However the forward evaluation is repeated at most as many number of times as the Linear Complexity of the sequence $\{y,F(y),\ldots\}$ to get one possible solution when this sequence is periodic. All other solutions are obtained in chains $\{e,F(e),\ldots\}$ for all points $e$ in the Garden of Eden (GOE) of the map $F$. Hence a solution $x$ exists iff either the former sequence is periodic or a solution occurs in a chain starting from a point in GOE. The online computation then turns out to be polynomial time $O(L^k)$ in the linear complexity $L$ of the sequence to compute one possible solution in a periodic orbit or $O(l)$ the chain length for a fixed $n$. Hence this is a complete algorithm for solving the problem of finding all rational solutions $x$ of the equation $F(x)=y$ for a given $y$ and a map $F$ in $\ff_{2^n}$. Due to the dependence on forward evaluation by $F$ only, in the online computation this algorithm is expected to be scalable for inversion in realistic cases of maps when the linear complexity $L$ of the sequence above or lengths of chains are not exponential in $n$. The chains emerging from the GOE are independent and can be computed in parallel. The hard (NP) offline computation consists of computing all points in the GOE and the set of all orbit lengths of closed orbits of the iterative action of $F$. 
\end{abstract}
\section{Introduction}
Inversion of a function is a fundamental problem of both pure and applied mathematics and has enormously wide implications and important applications. In this paper we consider the problem \emph{local inversion} of a function $F:\ff_{2^n}\mapsto\ff_{2^n}$ which is to find all rational (in $\ff_{2^n}$) solutions $x$ to the equation $F(x)=y$ for given $y$ in $\ff_{2^n}$. For instance, inversion of One Way Functions (OWFs) is a central problem of Cryptanalysis. Although such a problem can obviously be solved by an exponential time brute force search algorithm, the challenge is to develop algorithms which will be practically feasible for solving realistic cases of inversion arising in applications where brute force is infeasible. Another methodology to solve such a problem better than brute force approach is by means of a solver for Boolean equations which should return all solutions $x$. Such a solver in which all solutions are represented by a complete set of implicants of the system of equations is announced in \cite{Implicant}. In this paper we propose an algorithm which utilizes an offline computation on $F(x)$ while the online computation utilizes only the forward computation of the map $F(z)$ on points $z$ and finds all solutions using the sequence $S(F,z)=\{z,F(z),\ldots\}$. Since the forward computation $F(z)$ is assumed to be efficiently feasible, the estimates of online computation required are expected to be better than solving the Boolean system $F(x)=y$ directly. The offline computation however involves hard computation and is assumed to be feasible with sufficient time and memory space for some of the realistic cases of $F$. For such cases local inversion can be practically feasible if the online computation is feasible.  

An important application of local inversion arises in Cryptanalysis of cipher algorithms where it is necessary to find all inputs (arguments) to a known map $F$ for a given output (map value). For instance, the output is associated with a known ciphertext of a known or chosen plaintext while the unknown input is the secret key. The cipher algorithm is designed in such a way that the map $F$ is a OWF which is efficient for forward operation on its argument but hard to invert given an output. The well known probabilistic algorithm, Time Memory Tradeoff (TMTO) attack is often used to solve such an inversion problem. TMTO attack finds one inverse of a map $F:X\rightarrow X$ in the set $X$ of $n$-bit strings given an ouput $y=F(x)$ probabilistically. TMTO attack has exponential complexity (in $n$) and is equivalent to $O(\sqrt{N})$ order Birthday attack where $N=|X|=2^n$. Hence TMTO has a success probability of $63\%$ once an offline data is gathered which is also a problem of $O(\sqrt{N})$ order of computation. Apart from being probabilistic TMTO is not meant for computing all solutions of $F(x)=y$ nor does the method utilize the structure of the finite field possible for strings of $n$-bits. A complete algorithm on the other hand by definition succeeds with certainty in computing the solutions. Purpose of this paper is to propose a complete algorithm for the local inversion problem by utilizing the structure of the map $F$ as an $n$-tuple polynomial function of $n$ variables over $\ftwo$. An important outcome of the result of this paper is that when the local inverse for a given $y$ is unique (which is most likely in a practical situation) and the linear complexity of the sequence $S(F,y)$ is not too large, then a key recovery attack on the cipher algorithm is practically feasible. 

\subsection{Offline and online computations for local inversion}
We show in this paper that the computation of local inversion can be divided into offline and online phases. 
\begin{enumerate}
    \item The offline phase is used for computing the set of all possible periods of closed orbits of iterations of $F$ which may require an exponential order of computation in general. 
    \item Offline computation is also required for computing the \emph{Garden of Eden} of the map $F$ which is the set of all $y$ which are outside the range of $F$ (or those $y$ for which there is no rational solution to the equation $F(x)=y$). This computation is equivalent to computing all satisfying assignments of a Boolean system of equations.
    \item Then in online computation, one local inverse $x$ when it lies in a periodic orbit containing $y$ can be computed in polynomial order of complexity $O(L^k)$ where $L$ is the \emph{linear complexity} of the sequence $S(F,y)=\{y,F(y),F^{(2)}(y),\ldots\}$ under iterations of $F$. An upper bound on $L$ is the maximum period of possible orbits of $F$. 
    \item With the knowledge of GOEs of $F$, other inverses of $y=F(x)$ belong to the sequences under iterations $F^k(z)$ on points $z$ in the GOE, known as chains of $F$ which is a polynomial time online parallel computation independent for each point $z$ in GOE.
\end{enumerate}     
The most important advantage of this splitting of computation is that, since the period of the sequence $S(F,y)$ is the \emph{order} of the \emph{minimal polynomial} in $\ftwo[X]$ the polynomial ring over $\ftwo$, the unique solution of the local inverse in the periodic orbit of $y$ in the above sequence has a polynomial time solution in the degree of the minimal polynomial of the sequence (or its linear complexity $L$). All other solutions too are computed by the forward application of $F$. Hence we show that by taking into account the structure of the finite field and the function $F$, one possible solution of the local inversion can be computed in $O(L^k)$ once the offline computation is completed.

\subsection{Relations with previous literature}
Linear complexity and the minimal polynomial of sequences have been well studied in the Cryptology literature since long \cite{Rueppel,HBCrypto}. The Berlekamp Massey algorithm \cite{vonzerGa} computes the minimal polynomial efficiently. However linear complexity was discovered for modeling and predicting the sequence as generated by the output of a LFSR and hence predicting an exponentially long sequence from just $2L$ samples by a linear recurrence of a polynomial size subsequence. The computation of linear complexity $L$ of a sequence is also not feasible if $L$ is exponential comparable to the period. This paper shows another application of linear complexity and the minimal polynomial of the sequence $S(F,y)$, that of inverting the function $F(x)$ at its value $y$

Our results are further development of ideas on linear representation of maps $F$ using the Koopman operator and its applications recently announced in \cite{ramsule}. It is shown that one solution of the inverse is obtained in polynomial time in $O(L^k)$ where $L$ is the linear complexity of the sequence $S(F,y)$. Use of linear complexity of such sequences for inversion of $F$ is believed to be a new proposition of this paper and follows from the ideas of \cite{rasu}. Other known literature relates to this problem as follows.

\subsubsection{Inversion of permutation polynomials}
The problem of inversion of permutation polynomials in finite fields has a vast literature. References \cite{LidNid, HBFF, zhengww} are good representatives. However inversion of permutation polynomials does not result in a method for local inversion. Recently in \cite{rasu} a linear representation of functions in finite fields is proposed and it is shown that when $F$ (is a permutation function) with the global inverse $G$ then the inverse $G$ as well as its linear representation also follows from that of $F$. The linear representations of inverses are inverses of linear representations. However linear representation of $F$ does not solve the problem of local inversion when $F$ is not globally invertible (since $G$ does not exist).  Hence local inversion problem needs to be solved afresh.

\subsubsection{Relation with TMTO}
As discussed above, the Time Memory Tradeoff (TMTO) attack on OWFs is a well known algorithm for solving the local inversion of a map $F:X\mapsto X$ in $X=\{0,1\}^n$ given $y=F(x)$ in $X$. TMTO algorithm was proposed by Hellman in \cite{Hellman} and has been well researched over last several years as shown in \cite{HongSarkar, Hays, Gango}. TMTO algorithm is however probabilistic and does not utilize structure of the set $X$ as a vector space over a finite field and its success probability is based on the Birthday attack which has complexity of order $O(2^{n/2})$. In practice TMTO attack requires gathering of multiple large sized tables offline not just one to overcome merging of chains, false positive solutions and even failure to capture a solution. However the online search can be carried out in parallel which is advantageous.

\subsubsection{Relation with Boolean equation solving}
An algorithm for representing all solutions (satisfying assignments) of a Boolean systems of equations is proposed in \cite{Implicant}. This algorithm is completely thread parallel hence the time performance improves with increase in memory space. The problem of representing all solutions of Boolean equations is however of $\sharp$-P class whose special case, that of deciding solvability of $3$-CNF system is an NP complete problem. If the Boolean equations are represented in the form $y=F(x)$ for a known $y$ in $\ftwo^n$ with $n$-unknowns taking values in $\ftwo$, this method can be used for local inversion. Hence al solutions of inverses can be solved using the algorithm proposed in \cite{Implicant} better than brute force search. However if the offline computation on $F$ of the method proposed in this paper is feasible, the local inversion does not require solution of non-linear Boolean equations but only depends on forward computations of the function. Hence the algorithm proposed in this paper is a definite way to solve these hard problems even without solving the system of Boolean equations and only resorting to forward computation of the map $F$ if the offline computation is feasible.  

\subsubsection{Related work on local inversion}
Problem of solving rational roots of polynomial of systems of equations over finite fields is one of the well known hard problems of computation. This problem has been addressed in past by algebraic geometric and probabilistic algorithmic methods \cite{cafure, courtois}. The case of deciding the solvability of a multivariate quadratic system is an NP complete problem. Hence in general the local inversion problem is computationally hard. More specifically the problem of inversion of permutation polynomials requires symbolic arithmetic and results in an inverse function \cite{zhengww}. The problem of local inversion of a function $F(x)$ at a point $y=F(x)$ is thus a special counterpart of both of these problems. Any algorithm for finding roots of polynomial systems can solve all rational roots in $\fq$ of the equation $y=F(x)$ given $y$. Similarly if $F$ is a permutation of $\fq$ and $G$ is the global inverse permutation then $x=G(y)$ is the unique inverse image of $y$. However when $F$ is not a permutation these methods cannot solve the local inversion problem. For $\ff_{2^n}$ existence and computation of all rational roots of $y=F(x)$ can be accomplished by the Boolean system solver \cite{Implicant}. However no such algorithm is known for fields in other characteristics. 
\section{Structure of trajectories and solutions}
Solutions to the equation $F(x)=y$ are governed by the structure of trajectories of the iterations of the map $F:\ftwo^n\mapsto\ftwo^n$ and on what part of a trajectory $y$ belongs. A trajectory of the map $F$ through a point $z$ in $\ftwo^n$ is the sequence of iterations $\{z,F(z),F^{(2)}(z)\ldots, F^{(k)}(z),\ldots\}$. Due to finiteness of $\ff_{2^n}$ and invariance of $F$ with respect to $k$, it follows that there exist $0\leq l\leq k$ such that $F^{(k)}(z)=F^{(l)}(z)$ for any $z$. Hence every trajectory can be uniquely split into a periodic part 
\beq\label{periodicorbit}
P(z)=\{F^{(k)}(z),F^{(k+1)}(z),\ldots,F^{(k+m-1)}(z)\}
\eeq
such that 
\[
F^{(k)}(z)=F^{(k+m)}(z)
\]
where $m$ depends on $z$, called a \emph{periodic orbit} of \emph{period} $m$ and a unique segment of a \emph{chain} 
\beq\label{chain}
C(z)=\{z,F(z),\ldots,F^{(k-1)}(z)\}
\eeq
which merges with a periodic orbit or another chain. Points of a periodic orbit of length $m$ are called $m$-periodic points. The \emph{fixed points} of $F$ which are points $z$ such that $F(z)=z$ are $1$-periodic points. The set of all points $z$ which are not in the range of $F$ is called the \emph{Garden of Eden} of $F$. Following lemma follows from this structure of trajectories.

\begin{lemma}\label{strofsolns}
Given an arbitrary point $y$ in $\ftwo^n$ solutions of $F(x)=y$ are described by following possibilities
\begin{enumerate}
    \item $F(x)=y$ has no solution iff $y$ belongs to the GOE of $F$.
    \item $F(x)=y$ has solution in a periodic orbit $P$ iff $y$ belongs to $P$. Such a periodic orbit is unique. Given $y$ in a periodic orbit $P$ and the unique solution to $F(x)=y$ in $P$, all other possible solutions belong to the chains $f^{k}(z),k\geq 1$ for $z$ in the GOE of $F$. 
    \item If $y$ is neither in a periodic orbit nor in the GOE, then all solutions to $F(x)=y$ arise in some of the chains $F^{(k)}(z),k\geq 1$ for $z$ in the GOE. 
\end{enumerate}
\end{lemma}

\begin{proof}
From the definition of the periodic orbits (\ref{periodicorbit}) and chains (\ref{chain}) it follows that the space $\ftwo^n$ is partitioned by the action of $F$ in periodic orbits and segments of chains from points in GOE to a point which is mapped by $F$ to a periodic orbit. Hence given any point $y$ it is either on a unique chain segment or on a periodic orbit. If there is no solution to $F(x)=y$ then $y$ belongs to GOE conversely if $y$ is in GOE then there is no $x$ such that $F(x)=y$. If $y$ is in a periodic orbit $P$, there is unique predecessor $x$ in $P$ such that $F(x)=y$. Any other solution $x$ which is outside $P$ cannot be in any other periodic orbit since the two orbits cannot intersect at $y$. Hence all other solutions $x$ are on chains merging with $P$ at $y$ under iterations of $F$. Hence every solution outside $P$ is on one of the chains $F^{(k)}(z)$ for $z$ in GOE and $k\geq 1$. 
\end{proof}

\subsection{Solution in a periodic orbit}
We need to revisit some of the well known background on the theory of recurrent sequences over finite fields \cite{LidNid, HBFF} for the special iterative sequences we encounter in this problem of local inversion of maps. The theory of linear recurring sequences over finite fields in \cite{LidNid, HBFF} considers recurrence relations with for sequences and their co-efficients in the same field. On the other hand the sequences encountered in the problem of local inversion are in $\ftwo^n$ and generated iteratively by the map $F$ while the co-efficients are over $\ftwo$. Because of this minor difference we choose to revisit proofs of fundamental results for the special sequences and show that the well known results go though without change. Consider a periodic orbit of iterations of $F$ which contains $y$, denoted as
\[
S(F,y)=\{y,F(y),F^{(2)}(y),\ldots,F^{(N-1)}(y)\}
\]
which has period $N$ which is the smallest number such that $F^{k}(y)=F^{(k+N}(y)$ for $k=0,1,2,\ldots$. We say that the sequence $S(F,y)$ has a \emph{linear recurrence relation} over $\ftwo$ of degree $m$ if there exist constants $\{a_0,a_1,\ldots,a_{(m-1)},a_m\}$ in $\ftwo$ such that
\beq\label{LRR}
a_mF^{(k+m)}(y)+\sum_{i=0}^{(m-1)}a_iF^{(k+i)}(y)=0
\eeq
The periodicity condition on a sequence $S(F,y)$ of period $N$,
\[
F^{(k+N)}(y)=F^{(k)}(y)
\]
gives existence of one such linear recurrence relation with $m=N$. The polynomial
\beq\label{charpoly}
\phi(X)=a_mX^m+\sum_{i=0}^{(m-1)}a_iX^i
\eeq
associated with a linear recurrence relation is called a \emph{characteristic polynomial} of $S(F,y)$. The indeterminate $X$ of a characteristic polynomial represents the operation
\beq\label{operationX}
X^{k}(y)=F^{(k)}(y)
\eeq
($X^0(y)$ is the identity map), which is also compatible with the compositional operation by $F$ as seen by
\[
X^{(k+l)}(y)=F^{(k+l)}(y)=F^{(k)}(F^{(l)}(y))
\]
The linear operation of terms in powers of $X$ is by definition
\[
(X^k+X^l)(y)=X^k(y)+x^l(y)
\]
A monic linear recurrence relation i.e. with $a_m=1$ and of least degree is unique and the monic polynomial $\phi(X)$ associated with it is called the \emph{minimal polynomial}. In a finite field $\fq$ a polynomial $\phi(X)$ in the polynomial ring $\fq[X]$ satisfying $\phi(0)\neq 0$ has an \emph{order}, which is the minimum number $N$ such that $\phi(X)$ is a divisor of $X^N-1$.

Following result follows from this background.

\begin{prop}\label{Existenceofmp}
The sequence $S(F,y)$ is periodic iff it has a minimal polynomial which divides any of its characteristic polynomials. The minimal polynomial satisfies $\al_0\neq 0$ and the period $N$ of $S(F,y)$ is the order of the minimal polynomial.
\end{prop}

\begin{proof}
Let $S(F,y)$ be periodic of period $N$, then it has the characteristic polynomial $\phi(X)=X^N-1$. Hence it has a minimal polynomial. Let the minimal polynomial of $S(F,y)$ be,
\beq\label{minpoly}
m(X)=X^m+\sum_{i=0}^{(m-1)}\al_iX^i
\eeq
Then from the definition of the operation $X$ in (\ref{operationX}) and interpretation of $X^k$ as compositional operation $F^{(k)}(.)$, $m(X)$ is the smallest degree polynomial such that
\[
m(X)(y)\triangleq F^{k+m}(y)+\sum_{i=0}^{(m-1)}\al_iF^{(k+i)}(y)=0
\]
If $\phi(X)$ is any characteristic polynomial and $R(X)$ the residue $\phi(X)\mod m(X)$ then there exists a polynomial $Q(X)$ such that $\phi(X)=Q(X)m(X)+R(X)$. Hence using the above algebraic rules of operation of $X$ on $y$ through the action of $F$ it follows that
\[
\phi(X)y=R(X)y=0
\]
hence due to minimality of degree of $m(X)$ and noting that $\deg R(X)<\deg m(X)$ it follows that $R(X)$ must be a zero polynomial. Hence in particular if $S(F,y)$ is periodic of period $N$ then $m(X)|(X^N-1)$. It also follows from this divisibility by $m(X)$ that $m(0)=\al_0\neq 0$. Moreover since $N$ is the smallest, $N$ must be the order of $m(X)$ as an element of $\ftwo[X]$. Conversely if $S(F,y)$ has minimal recurrence relation and the minimal polynomial satisfies $m(0)\neq 0$ then for $N$ the order of $m(X)$ it is a divisor of $X^N-1$. Hence $(X^N-1)(y)=0$ which proves that $S(F,y)$ is periodic of period $N$.
\end{proof}

\subsubsection{Rank criterion and computation of minimal polynomial}\label{sec:rankcriterion}
Consider the sequence $S(F,y)$. It is a-priori not possible to know whether the sequence is periodic from a short length samples $F^{(k)}(y)$ (for small $k$) since the period can be exponentially large. The issue here is that while the sequence $S(F,y)$ is not fully available for computation, to compute a linear recurrence relation (\ref{LRR}) of shortest degree. In fact this is what is the well known computation of linear complexity of a sequence by the Berlekamp-Massey algorithm. Important point is that such a computation requires the a-priori knowledge of the period of the sequence $S(F,y)$ to stop the search for the degree of the minimal polynomial. We introduce the Hankel matrix
\beq\label{hankelmatrix}
H_{m+j}=
\barr{llll}
y & F^{(j+1)}(y) & \ldots & F^{(j+(m-1))}(y)\\
F(y) & F^{(2)}(y) & \ldots & F^{(j+m)}(y)\\
\dots & \vdots &  & \vdots\\
F^{(j+(m-1)}(y) & F^{(j+m)}(y) &\ldots & F^{(j+(2m-2))}(y)
\earr
\eeq
Following proposition is well known on computation of minimal polynomial of sequences in finite fields. In the present context the sequence is the vector sequence $S(F,y)$ in $\ftwo^n$. Hence the Hankel matrix (\ref{hankelmatrix}) is no more square as in the case of sequences over finite fields but is a block Hankel matrix.

\begin{prop}
Let $S(F,y)$ be periodic then it has minimal polynomial of degree $m$ iff
\beq\label{rankcond}
\rank H_{m+j}=\rank H_{m}=m
\eeq
for all $j=0,1,2,\dots$. The co-efficient vector 
\[
\hat{\alpha}=(\al_0,\ldots,\al_{(m-1)})^T
\]
of the minimal polynomial is the unique solution of
\beq\label{Hankeleqn}
H_m\hat{\alpha}=h(m)
\eeq
where
\[
h(m)=[F^m(y),F^{(m+1)}(y),\ldots,F^{(2m-1)}(y)]^T
\]
\end{prop}

\begin{proof}
In Proposition \ref{Existenceofmp} it is shown that the sequence $S(F,y)$ is periodic iff it has a minimal polynomial. Hence it is required to find its degree and co-efficients using the sequence. Let $m(X)$ denoted in (\ref{minpoly}) be the minimal polynomial of degree $m$. Then there is recurrence relation satisfied by the sequence $S(F,y)$
\[
F^{m+j}(y)+\sum_{i=0}^{(m-1)}\al_iF^{(k+i+j)}(y)=0
\]
writing these equations for $j=0,1,2,\ldots$ and noting that there is a unique solution to the co-effcients $\al_i$ is equivalent to the rank condition (\ref{rankcond}) and the co-efficient vector is the solution of the linear system (\ref{Hankeleqn}).

Conversely if the rank condition is satisfied there is a shortest linear recurrence relation of degree $m$ satisfied by the sequence $S(F,y)$. The co-efficients of the recurrence are obtained uniquely by the linear system (\ref{Hankeleqn}).
\end{proof}

\subsubsection{Algorithm to compute a solution in a periodic orbit}
One solution of $F(x)=y$ when sequence $S(F,y)$ is periodic, is obtained as described in the theorem below.

\begin{theorem}\label{Solutioninperiodicorbit}
Let $S(F,y)$ be a periodic sequence and $m(X)$ as described in (\ref{minpoly}) be its minimal polynomial. Then there is a unique solution to $F(x)=y$ in $S(F,y)$ given by
\beq\label{solution}
x=(1/\al_0)[F^{(m-1)}(y)+\sum_{i=1}^{(m-1)}\al_iF^{(i-1)}(y)]
\eeq
\end{theorem}
\begin{proof}
Let $S(F,y)$ be periodic of period $N$. The point $x=F^{(N-1)}(y)$ then satisfies the equation $F(x)=y$, hence this is one and the unique solution of the equation in the periodic orbit. But then it follows that for this solution $x$, the periodic sequence $S(F,x)=S(F,y)$. Hence if $m(X)$ as described in (\ref{minpoly}) is the minimal polynomial of $S(F,Y)$,
\[
m(X)(x)=F^{(m)}(x)+\sum_{i=0}^{(m-1)}\al_iF^{i}(x)=0
\]
From this expression of $m(X)(x)$ the term $x$ can be solved uniquely since $\al_0\neq 0$. 
\[
x=(1/\al_0)[F^{(m)}(x)+\sum_{i=1}^{(m-1)}\al_iF^{(i)}(x)]
\]
Then by using the condition $y=F(x)$ one gets the relation (\ref{solution}). This is the expression of the unique solution $x$ in the periodic orbit $S(F,y)$.
\end{proof}
We note following important points about the solution (\ref{solution}).
\begin{enumerate}
    \item The unique solution (\ref{solution}) in the periodic orbit $S(F,y)$ is expressed in terms of the linear combinations of the sequence $S(F,y)$ which consists of only forward operations of the map $F$ (which is assumed to be feasible in practical time).
    \item The number of terms in (\ref{solution}) are $m$ the degree of the minimal polynomial. Hence the solution $x$ is computed in polynomial order of computations in $O(m^k)$ given $y$. 
    \item Since $S(F,y)$ is periodic of period $N$ if it has a minimal polynomial, $x=F^{(N-1)}(y)$ is also the same solution as computed in (\ref{solution}). However $N$ may be exponentially large hence this formula for computation of $x$ is not feasible without knowing the period $N$.
    \item If a solution $x$ is obtained by a possible minimal polynomial as in (\ref{solution}) then knowing the possible period $N$ it can be easily verified whether $x$ lies in the periodic orbit $S(F,y)$ by $F(x)=y$ instead of verifying by the equation $F^{(N-1)}(y)=x$. 
    \item If the degree $m$ itself is of polynomial order in $n$ (number of variables) then the online computation of the solution $x$ in $S(F,y)$ is computable in polynomial time. (The situation is comparable to primitive polynomials of degree $n$ whose order is $N=2^n-1$. Hence $m$ is expected to be of the order $\log N$ in same proportions (density) as primitive polynomials of degree $n$ among all polynomials). 
\end{enumerate}
Based on this theorem the online algorithm to compute one solution $x$ when $y$ is in a periodic orbit $S(F,Y)$ is Algorithm \ref{Solninperiodic}.

\begin{algorithm}
\caption{Online Algorithm: Unique solution in a periodic orbit}\label{Solninperiodic}
\begin{algorithmic}[1]
\Procedure{SolutioninPriodicorbit}{$S(F,y)$}
\State\textbf{Input}: Set of all possible periods $\hat{\Pi}$ (computed offline), $N_m=\mbox{max}\{\hat{\Pi}\}$, $m_0$ a miniumum value of degree below which minimal polynomial does not exist.
\State\textbf{Output}: One solution of $F(x)=y$ in the periodic orbit $S(F,y)$ if one exists otherwise returns that there is no solution in a periodic orbit.
\State Set $m=m_0$ 
\Repeat
\State Compute matrices $H_m$, $H_{(m+1)}$ 
\If{
\[
\rank H_m=\rank H_{m+1}=m
\]
}
    \State find solution $\hat{\alpha}$ (co-efficients of polynomial $m(X)$) in (\ref{Hankeleqn}).
    \For{$N\in\hat{\Pi}$}
    \While{$(m(X)|X^N-1)\bigwedge(F(x)=y)$}
    \State Return solution $x$ as computed in (\ref{solution})
    \State \% Solution in a periodic orbit of period $N$
    \EndWhile
\EndFor
\Else 
\State $m\leftarrow m+1$
\EndIf
\Until{$m>N_{m}$}
\State Return: No solution in periodic orbits.
\EndProcedure
\end{algorithmic}
\end{algorithm}

\subsubsection{Online Incomplete algorithm to compute local inverse}
It is interesting to observe how an online incomplete online algorithm can be constructed which may possibly find a solution given $y$ but without any offline computation. Such an algorithm may be practically useful within certain bounds on the degree of the minimal polynomial.
Let a bound $M$ on the degree of the minimal polynomial $m(X)$ is specified and the step of computation of $m(X)$ is stopped when degree $m$ of a possible minimal polynomial exceeds $M$. The bound $M$ should be theoretically at most equal to any of the possible periods of closed orbits $N$ of $F$. However $M$ is chosen from a point of view of feasibility of the computation of the minimal polynomial using the rank criterion of the the section \ref{sec:rankcriterion}. Hence $M$ may be chosen of polynomial order $O(n^k)$ in the number of variables for suitable $k$. Algorithm \ref{Incomplete} shows this computation.

\begin{algorithm}
\caption{Online Algorithm (Incomplete): Unique solution in a periodic orbit or report no conclusion}\label{Incomplete}
\begin{algorithmic}[1]
\Procedure{PossibleSolutionin}{$S(F,y)$}
\State\textbf{Input}: $y$, $m_0$ a miniumum value of degree below which minimal polynomial does not exist and $M$ maximum value.
\State\textbf{Output}: One solution of $F(x)=y$ in the periodic orbit $S(F,y)$ if one exists with minimal polynomial of degree $m<M$ otherwise returns that there is no conclusion.
\State Set $m=m_0$ 
\Repeat
\State Compute matrices $H_m$, $H_{(m+1)}$ 
\If{
\[
\rank H_m=\rank H_{m+1}=m
\]
}
    \State find solution $\hat{\alpha}$ (co-efficients of polynomial $m(X)$) in (\ref{Hankeleqn})
    \State compute $x$ as a possible solution as in (\ref{solution})
.
    \State compute $N=\mbox{ord}\;m(X)$
    \While{$F(x)=y$}
    \State Return solution $x$ as computed in (\ref{solution})
    \State \% Solution in a periodic orbit of period $N$
    \EndWhile
\Else 
\State $m\leftarrow m+1$
\EndIf
\Until{$m>M$}
\State Return: No conclusion about solution in a periodic orbit.
\EndProcedure
\end{algorithmic}
\end{algorithm}

Algorithm \ref{Incomplete} is practically feasible and results in one correct solution when the linear complexity of $S(F,y)$ is less than $M$. The algorithm eliminates false positive solutions by checking the equation $F(x)=y$ in step $11$.

Theorem \ref{Solutioninperiodicorbit} leaves open the question whether there are other solutions outside the periodic orbit $S(F,y)$ or when such solutions exist when $S(F,y)$ is not periodic. This problem is considered next.

\subsection{Solutions outside periodic orbits}
From the lemma \ref{strofsolns} it follows that solutions which are not in the periodic orbit $S(F,y)$ or if $S(F,y)$ is itself not periodic can only possibly arise on chains $F^k(z)$ starting from $z$ in GOE of $F$. Hence assuming the GOE is computed in offline computation an online algorithm for discovering solutions on chains from GOE uses only forward computations of the map $F$. However an efficient termination condition for the forward computation on a chain is needed.

\subsubsection{Termination of a chain on a periodic orbit and all solutions on the chains}
If $z(k)=F^{(k)}(z)$ is a sequence of points on a chain, the chain terminates when $z(k)$ meets a periodic orbit. If all possible periods of periodic orbits of $F$ are known a-priori then for each $N$ of these periods, the relation $F^{N}(z(k))=z(k)$ seems to be the only way to check whether $z(k)$ is on a periodic orbit of period $N$. Hence it is necessary to show that even when $N$ is exponential order in $n$, evaluation $F^N(z)$ can be carried out efficiently in polynomial time in $\log N$ on any point $z$ in $\ftwo^n$. This is shown in Appendix. Algorithm \ref{solutioninachain} to compute solutions of $F(x)=y$ outside periodic orbits $S(F,y)$ can be as follows.

\begin{algorithm}
\caption{Solution in a chain}\label{solutioninachain}
\begin{algorithmic}[1]
\Procedure{Solutioninachain}{z,y}
\State\textbf{Input}:
\State 1. $z$ a point in GOE (computed offline), 
\State 2. $y$ online value of $F(x)$ to be inverted, 
\State 3. set of all possible periods $\hat{\Pi}$ (computed offline)
\State\textbf{Output}: Solution of $F(x)=y$ on a chain starting from $z$.
\For {$N$ in $\hat{\Pi}$}
\State $k=1$,
\State Compute $z(k)=F^{(k)}(z)$
\While {$z(k)\neq y$}
\If {$F^{(N)}(z(k))=z(k)$}
\State Print: No solution on the chain from $z$.
\EndIf
\EndWhile
\If{$z(k)=y$}
\State a solution is 
\State $x=z(k-1)=F^{(k-1)}(z)$
\EndIf
\State $k\leftarrow k+1$
\State \Comment{The loop always ends because $z(k)$ joins a}
\State \Comment{periodic orbit of period $N$ or is equal to $y$}
\EndFor
\EndProcedure
\end{algorithmic}
\end{algorithm}

Algorithm \ref{solutioninachain} is repeated in parallel for all points $z$ in GOE.
Both the online algorithms Algorithm \ref{Solninperiodic} for computing an inverse in a periodic orbit containing $y$ and Algorithm \ref{solutioninachain} to compute all other solutions, require an a-priori (offline) computation on $F$ which is independent of $y$. This computation creates a data, the set $\hat{\Pi}$ of all possible periods of periodic orbits of $F$ and the set GOE. This offline computation is discussed in the next section.
\section{Computation of possible periods of orbits and the GOE}
For linear maps $L:\ff_{2}^n\mapsto\ff_{2}^n$, \cite{Gill} showed how the computation of the set of all periods of closed orbits and the set of all points in the GOE is accomplished by computing the rational canonical form of $L$ over $\ftwo$ in a fixed basis in $\ftwo^n$. The online algorithms presented in the previous section to compute local inverse of $F$ at $y$ make use of the offline data of all possible periods of closed orbits of the map $F$ and the GOE of $F$. However since $F$ is not linear the linear algebraic computation of \cite{Gill} is no more applicable. In this section we take up the problem of computation of the GOE when $F$ is not linear. We show that the computations of all periods of closed orbits of the map $F$ can be achieved by computing an analogous linear algebraic map whose set of periods of closed orbits contains this information of periods of orbits of $F$. The GOE of $F$ however needs to be computed separately by solving a Boolean system. These two computations are achieved as follows. 
\begin{enumerate}
    \item Computation of all possible periods of orbits of map $F$, by showing that this set belongs to the set of periods of all orbits of a linear map dual to $F$ restricted to a dual vector space. This observation allows the linear algebraic computation of all possible periods by the rational canonical form as shown in \cite{Gill}.
    \item Computation of GOE of the map $F$ by means of the implicant based Boolean solver algorithm of \cite{Implicant}. 
\end{enumerate}
Both of these computations are expected to be hard. However since these are performed offline on $F$ they may be feasible practically for special instances of $F$ given sufficient memory for parallel computation.

\subsection{Linear representation of map $F$}
To compute the information on periods of closed orbits of the map $F$ a suitable representation of the action of the map $F$ in a linear space is logically appropriate. We shall revisit the essential ideas on this representation which have recently been announced in \cite{ramsule}. The map $F$ acts in the space $\ftwo^n$ as a polynomial function of the co-ordinates of the space $\ftwo^n$. Let $V$ denote the $\ftwo$-linear space of $\ftwo$-valued functions on $\ftwo^n$. Denote by $\chi_i$ the co-ordinate functions in $V$ defined by their evaluation at points $x$ in $\ftwo^n$ as
\[
\chi_i(x)=x_i
\]
where $x_i$ is the $i$-th co-ordinate of $x$. The map $F$ has the dual action on functions $\phi$ in $V$ defined by the composition
\[
F^*(\phi(x))\triangleq \phi(F(x))
\]
$F^*$ is thus a linear operator $\phi\circ F$ in $V$. Consider the $F^*$-invariant subspace $W$ which is the smallest invariant subspace of $F^*$ containing all the co-ordinate functions $\chi_i$. The space $W$ can be obtained as the sum of cyclic invariant subspaces $W_i$ generated by $\chi_i$,
\beq\label{ithcyclicspace}
W_i=\mbox{span}_{\ftwo}\{\chi_i,F^*(\chi_i),\ldots,(F^*)^{(m_i-1)}(\chi_i)\}
\eeq
where $m_i$ is the smallest number such that $(F^*)^{m_i}(\chi_i)$ is linearly dependent on $W_i$. Define
\beq\label{spaceW}
W=W_1+W_2+\ldots+W_n
\eeq
This space $W$ is thus the smallest invariant subspace of $F^*$ which contains all co-ordinate functions. Consider the restriction $F_1=F^*|W$. We define a linear dynamic system using the action of $F_1$ in $W$ as shown next.

\subsection{Computation of all periods from the linear representation}
In this section we gather some of the relevant results of \cite{ramsule}. Consider a basis ${\cal B}$ (of functions in $V$) of $W$ as the ordered set
\[
{\cal B}=\{\psi_1,\psi_2,\ldots,\psi_N\}
\]
where $N=\dim W$. Define the matrix representation of $F_1$ by the relation
\beq\label{matrixrep}
F_1[\psi_1,\ldots,\psi_N]=
[\psi_1,\ldots,\psi_N]K
\eeq
where $K$ is an $N\times N$ matrix over $\ftwo$. Columns of $K$ are the co-efficients $k_{ij}$ appearing in the expression
\[
F_1(\psi_i)=\sum_{j=1}^{N}k_{ij}\psi_j
\]
We first show how $F_1$ hence $K$ capture information about the map $F$ and its dynamical trajectories.

\begin{prop}
The map $F$ is a permutation of $\ftwo^n$ iff $F_1$ is a permutation, equivalently the matrix $K$ is non-singular
\end{prop}

\begin{proof}
Let $\psi$ be any function on $\ftwo^n$. Then $F^*(\psi)= 0$ implies $\psi(F(x))=0$ for all $x$. Hence $F$ is one to one implies $\psi(y)=0$ for all $y$ in $\ftwo^n$ i.e. $\psi$ is the zero function. Hence $F^*$ is also one to one. Hence its restriction $F_1$ in $W$ is also one to one and equivalently the matrix representation $K$ is non-singular.

Conversely, assume $F$ is not surjective, then there exists a point $\beta$ in the GOE of $F$ for which there is no solution $x$ such that $\beta=F(x)$.
Any function $\psi$ in $F^*(W)$ is of the form
\[
\psi=\sum_{i=1}^{N}a_i\psi_i\circ F
\]
where $\psi_i$ are functions in the basis ${\cal B}$. 
Then for any function $\psi$ in $F^*(W)$ as above, the value
\[
\sum_{i=1}^{N}a_i\psi_i(\beta)
\]
does not belong to image the image of $\psi$. In particular for a co-ordinate function $\chi_i$ in $W$ which has a representation in the basis as
\[
\chi_i=\sum_{i=1}^{N}\al_i\psi_i
\]
the value of $i$-th component $\beta_i=\chi_i(\beta)$ does not belong to the image of $F^*\chi_i$. Hence it follows that $\chi_i$ does not belong to $F^*(W)$. But since $W$ is an invariant subspace of $F^*$ which contains all co-ordinate functions $\chi_i$ it follows that $F^*(W)=F_1(W)$ does not contain the function $\chi_i$. Which implies that $F_1$ is not surjective. Equivalently $K$ is singular.
\end{proof}

Above proposition gives a necessary and sufficient condition for $F$ to be a permutation in terms of the matrix representation $K$ on the subspace $W$. Next we draw a correspondence between the trajectories of dynamics of $F$ with dynamics of $K$ in $\ftwo^N$.

\subsubsection{Embedding of trajectories in  trajectories of a Linear dynamic system}
Consider the linear dynamic systems, called hereafter as \emph{Koopman system} corresponding to $F$, defined by iterations by the linear map 
\beq\label{Koopman}
\begin{array}{rcl}
K^T:\ftwo^N & \mapsto & \ftwo^N\\
y & \mapsto & K^Ty
\end{array}
\eeq
For a trajectory
\[
x,F(x),F^{(2)}(x),\ldots
\]
of iterations of the map $F$ in $\ftwo^n$ define an embedding in the trajectories of Koopman system in $\ftwo^N$ corresponding to the basis ${\cal B}$ by
\beq\label{embeddingofpoint}
\begin{array}{lcl}
\ftwo^n & \mapsto & \ftwo^N\\
x & \mapsto & \hat{\psi}(x)
\end{array}
\eeq
where
\[
\hat{\psi}(x)=[\psi_1(x),\psi_2(x),\ldots,\psi_N(x)]^T
\]
Then a trajectory of the map $F$ in $\ftwo^n$ is embedded as a trajectory of the linear map (\ref{Koopman}) in $\ftwo^N$ by
\beq\label{embedding}
\begin{array}{lcl}
(x,F(x),F^{(2)}(x),\ldots) & \mapsto & (\hat{\psi}(x), \hat{\psi}(F(x)),
\hat{\psi}(F^{(2)}(x)),\dots)\\
 & = & (\hat{\psi}(x), F^*\hat{\psi}(x),
(F^*)^2\hat{\psi}(x),\dots)\\
 & = & (\hat{\psi}(x), K^T\hat{\psi}(x),(K^T)^2\hat{\psi}(x),\ldots)
\end{array}
\eeq

\begin{lemma}Following correspondence holds between trajectories of iterations of map $F$ and that of $K^T$.
\begin{enumerate}
\item To every closed orbit of the iteration of the map $F$ there corresponds a closed orbit of the Koopman system (iterations of the linear map (\ref{Koopman})) of the same period
\item To every chain in iterations of $F$ there corresponds a chain in iterations of $K^T$ of same length.
\end{enumerate}
\end{lemma}

\begin{proof}
The proof essentially follows from the embeddings (\ref{embedding}) of trajectories of $F$ into that of trajectories of $K^T$. If there is a periodic trajectory $F^{(k)}(x)$ of $F$ of period $N$ then it satisfies $F^{(k+N)}(x)=F^{(k)}(x)$. Hence under the embeddings (\ref{embeddingofpoint}) and (\ref{embedding}) the trajectory satisfies
\[
(K^T)^{(k+N)}(\hat{\psi}(x))=(K^T)^{k}(\hat{\psi}(x))
\]
Since $N$ is the period of the trajectory it is also the period of the embedded trajectory.

The argument about the embedding of chain and their lengths also follows on similar reasoning.
\end{proof}

From this lemma we get a set of all numbers which contains all periods of closed orbits of $F$ as follows.

\begin{theorem}\label{periodembedding}
The set $\Pi$ of all periods of closed orbits of iterations of $F$ in $\ftwo^n$ is a subset of the set of all periods $\hat{\Pi}$ of closed orbits of iterations of $K^T$ in $\ftwo^N$. $\hat{\Pi}$ can be computed in polynomial time in $N$.
\end{theorem}

\begin{proof}
That $\Pi\subset\hat{\Pi}$ follows from the above lemma. The rational canonical form of $K^T$ gives all elementary divisors of $K^T$ and as shown in \cite{Gill} the periods of closed orbits of iterations of $K^T$ are obtained as the orders of powers of elementary divisors of $K^T$. Hence computation of $\hat{\Pi}$ can be accomplished in polynomial time in $N$ which is the size of $K^T$.
\end{proof}

For the sake of clarity of the computation of $\hat{\Pi}$, the set of periods of all closed orbits of iterations of $K^T$, we present the well known result from \cite{Gill} here without proof.

\begin{theorem}\label{Rationalformandperiods}
Consider the linear map $A:\ftwo^n\mapsto\ftwo^n$. Let 
\[
\{d_1(X),d_2(X),\ldots,d_w(X)\}
\]
denote the elementary divisors of $A$ in $\ftwo[X]$ such that $d_i(0)\neq 0$ and let each $d_i(X)$ has the form
\[
d_i(X)=(p_i(X))^{e_i}
\]
where $p_i(X)$ are irreducible polynomials in $\ftwo[X]$. Then the set of all periods of closed orbits of iteration of the map $A$ is given by
\[
\Pi(A)=\{1,T_{ij},i=1,\ldots,w, j=1,\ldots,e_i\}
\]
where
\[
\{T_{ij}=\mbox{Order}\;(p_i(X))^j,j=1,2,\ldots,e_i\}
\]
\end{theorem}

(Here $\mbox{Order}\;p(X)$ of a polynomial satisfying $p(0)\neq 0$ is the smallest number $T$ such that $p(X)$ divides $X^T-1$). For a proof of this well known theorem reader may see section 12.10 in \cite{Gill}. The theorem above shows that the computation of the set $\Pi(A)$ of periods of the map $A$ is accomplished by computing the rational canonical form of $A$ to compute all elementary divisors of $A$. 

\subsubsection{Offline algorithm for computing the set of all possible periods}
An offline algorithm can thus find the input $\hat{\Pi}$ required in Algorithms 1 and 2 in following by computing $\Pi(K)$ and treating it as $\hat{\Pi}$ since by theorem \ref{periodembedding} $\hat{\Pi}$ contains all possible periods of closed orbits of iterations of $F$.
\begin{enumerate}
\item Compute the matrix $K$, a representation of $F^*$ on the invariant space $W$ in a fixed basis.
\item Compute the set $\hat{\Pi}$ of all periods of closed orbits of $K$ by computing the rational canonical form of $K$ and using theorem (\ref{Rationalformandperiods}). (Note that this set is same as that of set of all periods of closed orbits of iterations of $K^T$). 
\end{enumerate}

The computation of the matrix representation $K$ of $F$ may turn out to be a hard computation of $NP$ class. This is because the dimensions of the subspaces $W_i$ in (\ref{ithcyclicspace}) are likely to grow very fast with $n$ depending on the degrees of polynomial components of the map $F$. Also the computation of composition $F^*(f)=f\circ F$ is complicated although theoretically of polynomial time in $n$. However as this computation is offline it has much relaxed conditions on practical time and memory issues as compared to online computation. However the computation of $\hat{\Pi}$ is polynomial time in $N$ the dimension of the space $W$. Doing this hard computation offline is the price one has to pay for getting a complete algorithm for inversion. 

\subsection{Computation of the GOE}
The input required for the Algorithm 2 is the set GOE of $F$. This computation is formalised as the following Boolean system problem.

\begin{problem}\label{GOEproblem}
Given the map $F$ in polynomial form, determine the set $G$ of all simultaneous assignments $\{y_i,i=1,\dots,n\}$ in $\ftwo$ such that for the vector $y=(y_1,\ldots,y_n)^T$ the Boolean system of equations
\[
F(x_1,\ldots,x_n)=y
\]
is not satisfiable for any assignments $\{x_i,i=1,\dots,n\}$ in $\ftwo$. Then GOE of $F$ is the set of all such vectors $y$ in $\ftwo^n$.
\end{problem}

Computation of GOE can be accomplished in offline by means of an algorithm which can represent all solutions of a Boolean system of equation. Such an algorithm is announced in \cite{Implicant}. The algorithm is in general of exponential order in $n$. However the algorithm is parallel and given sufficient memory has complexity $O(n)$ for parallel computation of an $n$ variable Boolean system. This computation is described in the next section. 

\section{Computation of the GOE in terms of implicants}
The final offline computation required to solve the local inversion problem is the computation of the GOE of the map $F$. As stated in Problem \ref{GOEproblem}, a strategy to solve this problem makes use of the Boolean system $F(X)=Y$ defined by the map $F$. Let the map $F$ be described in terms of its component functions of $n$ variables as
\[
F(x_1,\ldots,x_n)=[f_1(x_1,\ldots,x_n),\ldots,f_n(x_1,\ldots,x_n)]^T
\]
where $f_i$ are Boolean functions of $n$ Boolean variables variables $x_i$ taking values in the Boolean ring $\ftwo$. We shall briefly revisit background of the \emph{implicant} based representation of all solutions of Boolean systems from \cite{Implicant}. Let $g(x_1,\ldots,x_n)$ be a Boolean function of $X=\{x_1,\ldots,x_n\}$. A term $t(X)$ in Boolean variables $X$ is an elementary conjunction
\[
t(X)=\prod_{i\in\{1:n\}}x_i^{a_i}
\]
where $a_i=0,1$ and $x^a=x$ for $a=1$ while $x^a=x'$ (the compliment) if $a=0$. An implicant of a Boolean function $g(X)$ is a term $t(X)$ such that if $t(a)=1$ for any assignment $X=a$ for $a\in\ftwo^n$ then $g(a)=1$. A set $I_g$ of implicants of a function $g$ is said to be \emph{complete} if $g(a)=1$ for any assignment $X=a$ then there exists a term $t$ in $I_g$ such that $t(a)=1$. A basic fact in Boolean algebra is that if $I_g$ is a complete set of implicants of a function $g$ then
\[
g(X)=\sum_{t\in I_g}t(X)
\]
One of the fundamental problems of computation in Boolean algebra is the computation of the representation of all \emph{satisfying} assignments of a Boolean function $g$ of $n$ variables. An assignment $X=a$ is said to satisfy a Boolean function $g$ if $g(a)=1$. Such a representation is accomplished by an algorithm using implicants in \cite{Implicant}. This shall be described next.

\subsection{Representation of all satisfying assignments}
A set of nonzero \emph{orthogonal} (OG) terms $\{t_1,\ldots,t_m\}$ is one in which $t_it_j=0$ for $i\neq j$. An OG set of terms is said to be \emph{orthonormal} (ON) if
\[
\sum_{i}t_i(X)=1
\]
Hence following fact follows.

\begin{lemma}
If $I_g$ is a complete set of OG implicants of a Boolean function $g$ then the set of all satisfying assignments of $g$ is the set
\beq\label{satassofg}
S(g)=\cup_{t_i\in I_g}S(t_i)
\eeq
where $S(t_i)$ is the satisfying assignment of the implicant $t_i$.
\end{lemma}

\subsubsection{Implicants of a system}
For any implicant $t(X)$ the satisfying assignment $t(a)=1$ is a partial assignment of all literals in $t(X)$ to equal $1$. Hence the above expression (\ref{satassofg}) is a compact representation of all satisfying assignments of $g$. The union in (\ref{satassofg}) is disjoint because the implicants are OG hence for any $t_i=1$, $t_j=0$ at the same assignment for $j\neq i$.

The implicant based solver for a Boolean system of equations of \cite{Implicant} computes the set set of all satisfying assignments of a Boolean system in terms of a complete OG set of implicants of the system. Here an implicant of a system of Boolean functions $G=\{g_i(X)\}$, $i=1,\dots,m\}$ is meant a term $t(X)$ such that if $t(a)=1$ for an assignment $X=a$ then $g_i(a)=1$ for all $i=1,\ldots,m$. Hence such a system of Boolean functions is represented by a single function $g$ in terms of a complete set $I_G$ of implicants of the system $G$
\[
g(X)=\sum_{t\in I_G}t(X)
\]
Then we have
\[
g(X)=\prod_{i=1}^{m}g_i(X)
\]
Hence the set of all satisfying assignments of the system $G$ is given by 
\beq\label{allsatofG}
S(G)=S(g)=\sum_{t\in I_G}S(t)
\eeq

\subsection{Formulation of the problem of GOE}
Consider the map $F$ in Problem 1, then GOE of $F$ consists of points $(y_1,y_2,\ldots,y_n)$ in $\ftwo^n$ such that 
\[
F(x_1,\ldots,x_n)=(y_1,\ldots,y_n)^T
\]
has no satisfying assignments $X$. Consider the above system of Boolean equations in variables $X$ and $Y=(y_1,\ldots,y_n)$, denoted as $F(X)\oplus Y=0$. Let a complete set of OG implicants for satisfying assignments of this system of equations be denoted as $I(F,Y)$. Then each implicant in this set is of the form
\[
t(X,Y)=t_1(X)t_2(Y)
\]
where $t_1$ and $t_2$ are terms in $X$ and $Y$ respectively such that for assignments $t_1(a)=1$, $t_2(b)=1$, $(a,b)$ satisfies $t(a,b)=1$ for some $t$ in the set $I(F,Y)$. Define the function
\beq\label{functionphi}
\phi(Y)=\sum_{t\in I(F,Y)}t_2(Y)
\eeq

\begin{theorem}
GOE of $F$ is the set $S(\phi(Y)')$ or alternatively,
\[
\mbox{GOE}=\{Y|\phi(Y)=0\}
\]
\end{theorem}

\begin{proof}
Since $t$ is itself an implicant of the system $F(X)\oplus Y=0$, for an assignment $X=a,Y=b$ which satisfies this equation, $t_1(a)=t_2(b)=1$. From this it follows that the set of all assignments of $Y$ for which there exist assignments of $X$ of the system is obtained as satisfying assignments of
\[
\sum_{t\in I(F,Y)}t_2(Y)=1
\]
Hence all assignments of $Y$ for which there is no solution $X$ of the equation are equal to satisfying assignments of all the equations
\[
t_2(Y)=0,\forall t\in I(F,Y)
\]
This is expressed as the set of all solutions of the the equation
\[
\sum_{t\in I(F,Y)}t_2(Y)=0
\]
Defining the function $\phi(Y)$ as given the result follows.
\end{proof}

Such a set of assignments $Y$ can be thus computed by using the implicant algorithm again on the function $\phi(Y)'$.

\subsubsection{An example of computation of GOE using the Boolean solver algorithm}
Consider the map
\[
F(x_1,x_2,x_3)=[x_1\oplus x_2x_3\oplus x_1x_2x_3,x_1\oplus x_2,x_2\oplus x_1x_3]^T
\]
The system $F(X)\oplus Y=0$ above in this case is
\[
\begin{array}{rcl}
   g_1=x_1\oplus x_2x_3\oplus x_1x_2x_3\oplus y_1 & = &0  \\
    g_2=x_1\oplus x_2\oplus y_2 & = & 0\\
    g_3=x_2\oplus x_1x_3\oplus y_3 & = & 0
\end{array}
\]
which we shall denote as $S(X,Y)$. Define 
\[
g(X,Y)=(g_1(X,Y)\oplus 1)(g_2(X,Y)\oplus 1)(g_3(X,Y)\oplus 1)
\]
Then a complete set of implicants of $g(X,Y)$ is a complete set of implicants of the system $S(X,Y)$. We find this set using the Boolean system solver algorithm of \cite{Implicant}. Using the ON set $I=\{x_1,x_1'x_2,x_1'x_2'\}=\{t_i,i=1,2,3\}$ we find implicants of the function $g(X,Y)$.
\[
\begin{array}{llll}
\mbox{Pivote function $g$} & g/t_1 & g/t_2 & g/t_3  \\
g_2\oplus 1=1\oplus x_1\oplus x_2\oplus y_2 & x_2\oplus y_2 & y_2 & y_2'\\
\end{array}
\]
Hence $I(g_2\oplus 1)=\{x_1x_2'y_2,x_1x_2y_2',x_1'x_2y_2,x_1'x_2'y_2'\}$ denoted as $\{t_i,i=1,\dots,4\}$ next we use this implicant set
\[
\begin{array}{lllll}
\mbox{Pivote function $g$} & g/t_1 & g/t_2 & g/t_3 & g/t_4  \\
g_1\oplus 1=1\oplus x_1\oplus x_2x_3\oplus x_1x_2x_3\oplus y_1 & 
y_1 & y_1 & x_3\oplus y_1' & y_1\\
\end{array}
\]
Hence 
\[
I((g_1\oplus 1)(g_2\oplus 1))=\{x_1x_2'y_1y_2,x_1x_2y_1y_2',x_1'x_2y_2x_3y_1,x_1'x_2y_2x_3'y_1',x_1'x_2'y_1'y_2'\}
\]
This set has five implicants $\{t_i,i=1,ldots,5\}$. We next take the final pivot to find complete implicant set of the function $g$ or the system $S(X,Y)$.
\[
\begin{array}{llllll}
\mbox{Pivote function $g$} & g/t_1 & g/t_2 & g/t_3 & g/t_4 & g/t_5  \\
g_3\oplus 1=1\oplus x_2\oplus x_1x_3\oplus y_3 & x_3\oplus y_3' &
x_3\oplus y_3 & y_3 & y_3 & y_3'\\
\end{array}
\]
Hence the complete set of implicants of the system $F(X)=Y$ is
\[
\begin{array}{l}
\{x_1x_2'y_1y_2x_3'y_3',x_1x_2y_1y_2'x_3y_3',x_1'x_2x_3y_1y_2y_3,
x_1x_2'y_1y_2x_3y_3,x_1x_2y_1y_2'x_3'y_3,\\x_1'x_2x_3'y_1'y_2y_3,
x_1'x_2'y_1'y_2'y_3'\}
\end{array}
\]
Expressing these as products $\{t_{i1}(X)t_{i2}(Y)\}$ the factors affecting assignments of $Y$ for which there exist solution $X$ are
\[
T(Y)=\{y_1y_2y_3',y_1y_2'y_3',y_1y_2y_3,y_1y_2'y_3,y_1'y_2y_3,y_1'y_2'y_3'\}
\]
These are six minterms in $y_1,y_2,y_3$. Hence all points in GOE are given by solutions of
\[
\phi(Y)=\sum_{i}t_{i2}(Y)=0
\]
Since all minterms in $Y$ add up to $1$ the above equation is equivalent to the following equation in terms of minterms not present in the above set $T(Y)$,
\[
y_1'y_2y_3'+y_1'y_2'y_3=1
\]
On simplifying, this equation is equivalent to
\[
y_1'(y_2\oplus y_3)=1
\]
Hence $y_1=0$ along with pairs $(y_2=0,y_3=1)$, $(y_2=1,y_3=0)$. Hence GOE of $F$ consists of two points
\[
\mbox{GOE}=\{(0,0,1),(0,1,0)\}
\]
To verify correctness of the GOE computed above we can see actual trajectories of the map $F$ as follows
\[
\begin{array}{ll}
\mbox{orbit of length $1$} & (0,0,0)\mapsto (0,0,0) \\
\mbox{orbit of length $4$} & (1,0,0)\mapsto (1,1,0)\mapsto (1,0,1)
\mapsto (1,1,1)\mapsto (1,0,0)\\
\mbox{chain of length $2$} & (0,1,0)\mapsto (0,1,1)\mapsto (1,1,1)\\
\mbox{chain of length $1$} & (0,0,1)\mapsto (0,0,0)
\end{array}
\]
Hence the actual trajectories show that GOE consists of the predicted two points $(0,1,0),(0,0,1)$.
\section{Applications to Cryptanalysis}
In this section we briefly discuss application of the local inversion algorithm we have developed in previous sections to illustrate Cryptanalysis of block and stream ciphers.

\subsection{Cryptanalysis of block ciphers}
A block cipher is a map $E:\ftwo^{n+m}\mapsto \ftwo^p$ where $n$ is the length in bits of the symmetric key $K$, $m$ is the length of the input plaintext block $P$ and $p$ is the length of the output ciphertext $C$. $E$ is called an \emph{encryption} function and $C$ the encryption of $P$. The map $E$ satisfies the properties:
\begin{enumerate}
    \item Computation of $C$ for given $P$ and $K$ is efficiently feasible.
    \item Computation of $P$ given $C$ and $K$ is efficiently feasible. 
    \item For any fixed $K$, $E(K,X):\ftwo^m\mapsto\ftwo^p$ is a practical OWF. Alternatively this condition is same as, given a pair $(P,C)$ such that $C=E(K,P)$ for some $K$ it is practically infeasible to find any bit of $K$.
\end{enumerate}

The Cryptanalysis problem in its general sense is usually understood as recovery of $K$ given $(P,C)$ pair of blocks. Such a problem is practically justified because same key $K$ is used in encryption of many plaintext blocks $P$. Some of these blocks are known a-priori. Hence if a likely block $P$ is fixed the problem is to compute $K$ given $C=E(K,P)$. Hence this is a local inversion problem of $F(X)=E(X,P)$ when $n=m=p$ which we shall address here.

\subsubsection{Algorithm for cryptanalysis}
The theory of local inversion developed in previous sections leads to Algorithm 3 for cryptanalysis of block ciphers of above form with $n=m=p$.

\begin{algorithm}
\caption{Cryptanalysis of block ciphers $C=E(K,P)$}
\begin{algorithmic}[1]
\Procedure{BlockcipherCryptanalysis}{$C=E(K,P)$}
\State\textbf{Input}: The algorithm $E$ and one pair $(P,C)$.
\State\textbf{Output}: All solutions $K$ such that $C=E(K,P)$.
\State \emph{Offline Computation}: 
    \State Compute the set $\hat{\Pi}$ of all possible periods of closed orbits of the map $F(X)=E(X,P)$.
\State \emph{Offline Computation}: 
    \State Compute the set GOE of the map $F$.
\State \emph{Online Computation}: 
    \State Use Algorithm 1 to find the minimal polynomial of the sequence $S(F,C)$ and one solution if the sequence turns out periodic.
\State \emph{Online Computation}: 
    \State Find all other solutions using the Algorithm 2.
\EndProcedure
\end{algorithmic}
\end{algorithm}

The algorithm 3 shows that once the offline computation of periods of closed orbits and GOE of $F(X)=E(X,P)$ is carried out apriori, the online computation of all $K$ for a given $C$ can be achieved only by forward iterative action by $F$ which is computationally very efficient.

\subsection{Cryptanalysis of stream ciphers}
Stream ciphers come with an iteration map (a dynamical system) with an output function. A general model of such an algorithm is
\beq\label{streamcipher}
\begin{array}{rcl}
x(k+1) & = & F(x(k))\\
w(k) & = & f(x(k))
\end{array}
\eeq
$k=0,1,2,\dots$. The initial condition $x(0)=(K,IV)$ where $K$ is symmetric key. The output stream is $w(k)$ is used for encryption of the plaintext stream of bits $p(k)$ as
\beq\label{encryptionbyw(k)}
c(k)=p(k)\oplus w(k)
\eeq
Hence $w(k)$ is called \emph{keystream}. When some of the bits $p(k)$ of plaintext at instances $k$ are known, then the keystream bits $w(k)$ are also known. 

\subsubsection{Cryptanalysis problems}
Cryptanalysis of the stream cipher (\ref{streamcipher}), (\ref{encryptionbyw(k)}) consists of two problems
\begin{enumerate}
    \item Problem 1. (Internal state recovery). Given a partial key stream $w(k)$ over some interval $k_0,k_0+1,\ldots,k_0+t$, find all possible internal states $x(k_0)$.
    \item Problem 2. (Key recovery from internal state). Given an internal state $x(k_0)$ at some $k_0$ find all the initial states $x(0)$. This leads to key recovery by matching the known IV in the $x(0)$. 
\end{enumerate}
We show how each of these can be formulated as local inversion problems of maps.

\subsubsection{Internal state recovery}
Let $F^*$ denote the dual linear map on the space of functions $V$ on $\ftwo^n$ for the map $F$. For a keystream $w(k_0),\ldots,w(k_0+n-1)$, of length $n$, the internal state $x(k_0)$ is related by the equations
\beq\label{streamtointernalstate}
w(k_0+j)=(F^*)^jf(x(k_0))
\eeq
Hence when all the $n$ samples of $w(k_0+j)$ for $j=0,\ldots,(n-1)$ denoted as $\hat{w}(k_0)$ are given the equation (\ref{streamtointernalstate}) defines the map
\beq\label{internalstatemap}
\hat{F}(x(k_0))=
[f,F^*f,\ldots,(F^*)^{(n-1)}f]^T(x(k_0))
\eeq
Whose local inversion at $\hat{w}(k_0)$ gives solutions of internal states. Due to the special structure of the map $\hat{F}$ the evaluation of the sequence $S(\hat{F},\hat{w}(k_0))$ can be carried out by repeated forward action of the stream cipher map $F$ itself. Hence online computation by algorithms 1 and 2 is as easy as the stream cipher operation itself.

\subsubsection{Computation of set of possible periods of iteration of $\hat{F}$}
For local inversion of the map $\hat{F}$ it is thus necessary to examine the computation of the set of all possible periods of iterations of $\hat{F}$. This map $\hat{F}$ is same as given in (\ref{internalstatemap}) acting on any state $x$ in $\ftwo^n$. The definition of the dual map $\hat{F}^*$ is then
\beq\label{dualofhatF}
\begin{array}{lcl}
\hat{F}^*(\chi_1) & = & f\\
\hat{F}^*(\chi_2) & = & (F^*)f\\
\vdots & & \vdots\\
\hat{F}^*(\chi_n) & = & (F^*)^{(n-1)}f
\end{array}
\eeq
Alternatively the map $\hat{F}:\ftwo^n\mapsto\ftwo^n$ is defined as
\beq\label{maphatF}
\hat{F}=[f(x),(F^*f)(x),\ldots,(F^*f)^{(n-1)}(x)]^T
\eeq
Hence it follows that for any function $\phi(x_1,\ldots,x_n)$
\beq\label{dualmap}
\hat{F}^*\phi=\phi(f,F^*f,\ldots,(F^*)^{(n-1)}f)
\eeq

\begin{definition}
A stream cipher (\ref{streamcipher}) denoted as system $(F,f)$ is said to be observable if $\hat{F}$ defined in (\ref{maphatF}) is a permutation of $\ftwo^n$.
\end{definition}

\begin{prop}\label{observabilityandS(F,f)}
The stream cipher defined by the system $(F,f)$ is observable iff all co-ordinate functions belong to the space of functions
\beq\label{defofS}
S(F,f)=\mbox{span}\;\{f,F^*f,(F^*)^2f,\ldots,(F^*)^{(n-1)}f\}
\eeq
\end{prop}

\begin{proof}
Let the stream cipher $(F,f)$ be observable which by definition implies that $\hat{F}$ is a permutation of $\ftwo^n$. Let $\hat{G}$ be the inverse of $\hat{F}$ and let $\hat{G}^*$ denote its dual linear map on functions as defined in (\ref{dualmap}). By definition of inverse it follows that
\[
\hat{G}\circ\hat{F}=Id_{\ftwo^n}\Leftrightarrow\hat{F}^*\hat{G}^*=Id_{V}
\]
By the expression of $\hat{F}^*$ in (\ref{dualofhatF}) it follows that
\[
\chi_i=\hat{G}^*(F^*)^{(i-1)}f
\]
for all $i=1,\ldots,n$. Hence all co-ordinate functions belong to the space of functions $S(F,f)$.

Conversely, if all co-ordinate functions $\chi_i$ belong to $S(F,f)$, then the map $\hat{F}$ is surjective on $\ftwo^n$ hence for an output stream $w(k_0+j)$ of length $n$-there is a unique $x=x(k_0)$ such that equations (\ref{streamtointernalstate}) hold. This implies observability of $(F,f)$.
\end{proof}

We thus get the following observation about the set of all possible periods of closed orbits of iterations of the map $\hat{F}$.

\begin{theorem}
If the stream cipher system $(F,f)$ is observable, then 
\begin{enumerate}
\item $S(F,f)$ is the smallest $\hat{F}^*$-invariant subspace which contains all co-ordinate functions as well as the function $f$.
\item All periods of closed orbits of iterations of $\hat{F}$ are contained in the periods of closed orbits of iterations of the linear map
\[
\hat{F}_1=\hat{F}^*|S(F,f)
\]
whose matrix representation in a basis of $S(F,f)$ is denoted $K$. Hence all possible periods of closed orbits are contained in $\hat{\Pi}$ the set of all periods of closed orbits of $K$
\item Every output stream $w(k)$ of length $n$ has a unique solution to the internal state.
\end{enumerate}  
\end{theorem}

\begin{proof}
Let a function $g$ be in $S(F,f)$. Then by definition of $S(F,f)$ in (\ref{defofS}), $g$ has an expression
\[
g=\sum_{j=0}^{n-1}\alpha_j(F^*)^{j}f
\]
for some constants $\alpha_j$ in $\ftwo$. Now $\hat{F}^*g$ is
\[
g(\hat{F}\chi_1,\ldots,\hat{F}\chi_n)
\]
Hence using expressions (\ref{dualofhatF}) and expression of $g$ itself as linear combination of generators of $S(F,f)$ it follows that $\hat{F}^*g$ belongs $S(F,f)$. Which proves that $S(F,f)$ is $\hat{F}^*$-invariant. It is smallest such space being cyclically generated and contains both $f$ and all co-ordinate functions.

Being a permutation, all trajectories of iteration of $\hat{F}$ are closed orbits. Let $\dim S(F,f)=N$ then $K$ is an $N\times N$ matrix and we can consider trajectories of iteration of $K$ in $\ftwo^N$. From Proposition 2 it follows that $K$ is nonsingular hence trajectories of iteration of $K$ are also closed orbits. From Theorem \ref{periodembedding} it follows that the set of all possible periods of orbits of $\hat{F}$ are contained in the set of periods $\hat{\Pi}$ of closed orbits of $K$.

Unique internal state corresponds to an output stream of length $n$ since $\hat{F}$ is a permutation.
\end{proof}

The offline algorithm for computation of all possible periods of $\hat{F}$ follows from above theorem. We can thus state the solution of Problem 1 of cryptanalysis of stream ciphers as in Algorithm 4 below. Assume that $(F,f)$ is observable.

\begin{algorithm}
\caption{Cryptanalysis of stream ciphers $(F,f)$}
\begin{algorithmic}[1]
\Procedure{Streamciphercryptanalysis}{$w(k_0+k),k=0,\ldots,n$}
\State\textbf{Input}: The algorithm $(F,f)$ and sequence $\hat{w}$ of length $n$.
\State\textbf{Output}: Unique internal state $x(k_0)$.
\State \emph{Offline Computation}: 
\State Compute the set $\hat{\Pi}$ of all possible periods of closed orbits of the map $\hat{F}$ as the set $\hat{\Pi}$ of periods of closed orbits of the linear map $K$.
\State \emph{Online Computation}: 
\State Search for the minimal polynomial of the sequence $S(\hat{F},\hat{w})$. Verify that the order of the polynomial belongs to $\hat{\Pi}$. 
\State Find the unique solution $x(k_0)$ using formula in equation (\ref{solution})
\[
x(k_0)=(1/\alpha_0)[\hat{F}^{(m-1)}(\hat{w})+\sum_{i=1}^{(m-1)}\al_i\hat{F}^{(i-1)}(\hat{w})]
\]
\State The computation of $\hat{F}^{(k)}(w)$ is defined by single step $\hat{F}(w)$ as
\[
\hat{F}(w)=[f(w),f(F(w)),\ldots,f(F^{(n-1)}(w))]^T
\]
\EndProcedure
\end{algorithmic}
\end{algorithm}

Note that as $(F,f)$ is assumed observable the map $\hat{F}$ is a permutation. Hence the GOE is empty. The online computation only involves forward computation by the map $\hat{F}$ which is in turn obtained via forward computation of the stream cipher $(F,f)$ which is very efficiently possible. The expression of $x(k_0)$ can also be alternatively obtained by linear representation of the inverse map $\hat{G}$ computed from $\hat{F}_1$ as shown in \cite{rasu}. Yet another approach to recovering the internal state using observer theory is announced in \cite{rasu2}.

\subsubsection{Case of unobservable $(F,f)$}
In the general case of the stream cipher system $(F,f)$, the map $\hat{F}$ is not a permutation. Hence local inversion of the equation
\beq\label{Booleansystem}
\hat{F}(x)=\hat{w}
\eeq
to recover all internal states $x$ requires computation of the GOE of $\hat{F}$ in addition to the period of $S(\hat{F},\hat{w})$. Although this additional requirement (in offline) makes the computation of all internal states more complex in the unobservable case, this approach of computing one solution in a periodic orbit of $S(\hat{F},\hat{w})$ online by Algorithm 1 and all other solutions using the chains starting from GOE by Algorithm 2 is expected to be far faster than direct solution of the Boolean system (\ref{Booleansystem}). Hence at the cost of doing an offline computation using the Boolean solver for GOE the online computation is made much effcient.

\subsubsection{Key recovery from internal state: solution to Problem 2}
This problem has the fastest solution in the case when the map $F$ of the system is a permutation. Suppose $F$ is a permutation. Then by locating the inverse of the internal state $x(k_0)$ in the periodic orbit $S(F,x(k_0))$ the unique inverse can be found by repeated forward computation of $F$. Repeating this process till one gets $x(0)$ recovers the initial condition. Hence the offline computation involves only the computation of the set of all periods of closed orbits of iteration of $F$ by linear representation of $F$.

When $F$ is not a permutation, an expression of $x(k_0)$ when $k_0$ is a-priori known, can be obtained by offline computation as
\[
F^{(k_0)}(x_0)=x(k_0)
\]
This is again a local inversion problem of the map $F^{(k_0)}$. This problem can be solved by repeated forward computations on $x(k_0)$ and the GOE of the map as shown in Algorithms 1 and 2. Hence the offline computation consists of finding the linear representation of $F^{k_0}$ and set of all periods of closed orbits of its iteration. By linear representation of $F$, as the restriction of the dual linear map $F^*$ on the $F^*$-invariant subspace $W$, let $K$ denotes the matrix representation. Then the linear representation of $F^{(k_0)}$ is obtained as the matrix $K^{k_0}$. Unfortunately computation of GOE does not allow such a simplification in terms of $K$. 

\section*{Appendix}
We show here how the checking of the condition $F^{N}(y)=y$ is efficiently feasible in Algorithm \ref{solutioninachain}.
Consider the map $F$ in $\ftwo^n$ and assume that an evaluation $F(x)$ at any point $x$ is efficiently possible. We have following rules for composition of $F$ with itself $(F\circ F)(x)=F(F(x))$ for non-negative numbers $a,b$.
\[
\begin{array}{lcl}
F^{(a)}(x) & \triangleq & (F\circ\ldots\circ F)(x),a\;\mbox{times. By definition}\\
F^{(a+b)}(x) & = & (F^{(a)}\circ F^{(b)})(x)\\
F^{(ab)}(x) & = & (F^{(a)}\circ\ldots\circ F^{(a)})(x),b\;\mbox{times.}
\end{array}
\]
From these it follows that for the repeated squaring under composition we have
\[
F^{(2^i)}(x)=(F^{(2^{(i-1)})}\circ F^{(2^{(i-1)})})(x)
\]
Hence using this recursive rule $F^{(2^i)}(x)$ can be computed in $i$ compositions. Now to compute $F^{(N)}(x)$ for an exponential $N$, we can write binary expansion
\[
N=n_0+n_1*2+n_2*2^2+\ldots+n_i*2^{(i-1)}
\]
for $N$ of $\log N=i$. Then by above rules of composition,
\[
F^{N}(x)=F^{(n_0)}\circ F^{(n_1*2)}\circ\ldots\circ F^{(n_i)*2^{(i-1)}}(x)
\]
These are $i$ compositions and evaluations of functions. But each composition above can be achieved in at most $i$ compositions and evaluations. Since basic operation of evaluation $F(x)$ is efficiently feasible all of the above individual evaluations, the computation of $F^{N}(x)$ is feasible in polynomial number of evaluations in $\log N$.
\section{Conclusions}
A complete algorithm is constructed for local inversion of a map $F:\ftwo^n\mapsto\ftwo^n$ at $y$ in $\ftwo^n$ which is equivalent to an algorithm for computation of all solutions $x$ of the equation $F(X)=y$ over $\ftwo^n$. The algorithm is constructed in two parts one which depends only on $F$ called Offline computation and another which depends on $y$ and a repeated forward action of $F$ called online computation. Offline computation consists of computation of two sets associated with $F$, the set of all possible periods of closed orbits of iterations of $F$ and the GOE of $F$. The set of closed orbits of the a linear representation of $F$ is shown to contain all possible periods of closed orbits of $F$. It is shown that if there exists a solution in a periodic orbit $S(F,y)$ then this solution is unique in any periodic orbit and is efficiently obtained by forward computation by $F$ once the minimal polynomial of the orbit is computed. Finally other solutions of the inversion lying on chains are discovered by computing the chains by the forward map $F$ starting from the points in GOE. These computations are inherently parallel for each point in GOE. The most important practical impact of this paper for Cryptanalysis is that, one solution of the inverse is feasible efficiently for an online data when the linear complexity of $S(F,y)$ is small enough to facilitate practically feasible computation even if data of offline computation is not available. Offline computation is required to compute all other solutions. Hence any encryption algorithm in order to be secure must have high linear complexity of all of its periodic orbits. Since the offline computation can be efficiently done by parallel computation, an estimation of complete security of cipher algorithms by inversion is practically feasible and can be carried out offline by large scale parallel computation. 

\subsection*{Acknowledgements}
\noindent
    Author is thankful to Shashank Sule for his suggestions which led him to the investigation of the local inversion problem. He is also thankful to Ramachandran for pointing out important omissions in the first draft of the paper.

\end{document}